\documentclass[11pt]{article}

\usepackage[margin=1in]{geometry}
\usepackage{amsmath,array,bm}
\usepackage{amsmath,amsfonts}
\usepackage{amsthm}
\usepackage{amssymb,latexsym}
\usepackage{algorithm}
\usepackage{subcaption}
\usepackage{algorithmicx}
\usepackage[]{algpseudocode}
\usepackage{bbm}
\usepackage{mathrsfs}
\usepackage[dvipsnames]{xcolor}
\usepackage{graphicx}
\usepackage{caption}
\usepackage{enumitem}
\usepackage{float}
\usepackage{thmtools}
\usepackage{hyperref}
\usepackage{mathtools}
\usepackage{cleveref}
\usepackage{bussproofs}
\usepackage[square,sort,comma,numbers]{natbib}
\usepackage{tikz}
\usepackage{xspace}
\usepackage{xcolor}
\usepackage{marginnote}
\usepackage{tipa}
\usepackage{mdframed}
\usepackage[export]{adjustbox}[2011/08/13]
\usepackage{multirow}
\usepackage{hhline}
\usepackage{todonotes}
 \usepackage{tcolorbox}

\newtheorem{theorem}{Theorem}

\theoremstyle{definition}
\newtheorem{definition}{Definition}
\newtheorem{observation}{Observation}

\newtheorem{lemma}{Lemma}
\newtheorem{remark}{Remark}

\crefname{invar}{invariant}{invariants}
\crefname{ineq}{inequality}{inequalities}
\crefname{constr}{constraint}{constraints}
\crefname{tbl}{table}{tables}
\crefname{lem}{lemma}{lemmata}
\crefname{lemma}{lemma}{lemmata}
\crefname{cond}{condition}{conditions}

\title{Non-Uniform $k$-Center and Greedy Clustering}
\author{Tanmay Inamdar\\{ University of Bergen}\\{\small \texttt{Tanmay.Inamdar@uib.no}} \thanks{The first author is supported by the European Research Council (ERC) via grant LOPPRE, reference 819416.} \and Kasturi Varadarajan\\ { The University of Iowa}\\{\small \texttt{kasturi-varadarajan@uiowa.edu}}}
\date{}
\newcommand{\B}{\mathcal{B}}

\newcommand{\real}{\mathbb{R}}
\newcommand{\integer}{\mathbb{Z}}
\renewcommand{\natural}{\mathbb{N}}

\newcommand{\Poly}{\mathscr{P}}
\newcommand{\F}{\mathscr{F}}
\newcommand{\cov}{\mathsf{cov}}
\newcommand{\covhat}{\hat{\mathsf{cov}}}
\newcommand{\Child}{\mathsf{Child}}
\newcommand{\I}{\mathcal{I}}
\newcommand{\polycov}{\Poly^{\I}_{\cov}}
\newcommand{\T}{\mathcal{T}}

\renewcommand{\S}{\mathcal{S}}

\newcommand{\Q}{\mathcal{Q}}
\newcommand{\nukc}{NU$k$C\xspace}

\ifdefined\C
\renewcommand{\C}{\mathcal{C}}
\else
\newcommand{\C}{\mathcal{C}}
\fi
\newcommand{\red}[1]{{\color{red} #1}}
\newcommand{\lr}[1]{\left( #1\right)}
\newcommand{\LR}[1]{\left\{ #1\right\}}

\newcommand{\NP}{\textsf{NP}}
\newcommand{\Red}{\ensuremath{\mathsf{Red}}}

\newcommand{\Blue}{\ensuremath{\mathsf{Blue}}}

\newcommand{\ball}{\ensuremath{\mathsf{Ball}}}

\newcommand{\bhat}{\widehat{\mathcal{B}}}

\begin{document}

\maketitle

\begin{abstract}
	In the Non-Uniform $k$-Center (\nukc) problem, a generalization of the famous $k$-center clustering problem, we want to cover the given set of points in a metric space by finding a placement of balls with specified radii. In $t$-\nukc, we assume that the number of distinct radii is equal to $t$, and we are allowed to use $k_i$ balls of radius $r_i$, for $1 \le i \le t$. This problem was introduced by Chakrabarty et al. [ACM Trans.\ Alg.\ 16(4):46:1-46:19], who showed that a constant approximation for $t$-\nukc is not possible if $t$ is unbounded. On the other hand, they gave a bicriteria approximation that violates the number of allowed balls as well as the given radii by a constant factor. They also conjectured that a constant approximation for $t$-\nukc should be possible if $t$ is a fixed constant. Since then, there has been steady progress towards resolving this conjecture -- currently, a constant approximation for $3$-\nukc is known via the results of Chakrabarty and Negahbani [IPCO 2021], and Jia et al.\ [To appear in SOSA 2022]. We push the horizon by giving an $O(1)$-approximation for the Non-Uniform $k$-Center for $4$ distinct types of radii. Our result is obtained via a novel combination of tools and techniques from the $k$-center literature, which also demonstrates that the different generalizations of $k$-center involving non-uniform radii, and multiple coverage constraints (i.e., \emph{colorful $k$-center}), are closely interlinked with each other. We hope that our ideas will contribute towards a deeper understanding of the $t$-\nukc problem, eventually bringing us closer to the resolution of the CGK conjecture.
\end{abstract}

\section{Introduction}

The $k$-center problem is one of the most fundamental problems in clustering. The input to the $k$-center problem consists of a finite metric space $(X, d)$, where $X$ is a set of $n$ points, and $d: X \times X \to \real^+$ is the associated distance function satisfying triangle inequality. We are also given a parameter $k$, where $1 \le k \le n$. A solution to the $k$-center problem consists of a set $C \subseteq X$ of size at most $k$, and the cost of this solution is $\max_{p \in X} d(p, C)$, i.e., the maximum distance of a point to its nearest center in $C$. Alternatively, a solution can be thought of as a set of $k$ balls of radius $\max_{p \in X} d(p, C)$, centered around points in $C$, that covers the entire set of points $X$. The goal is to find a solution of smallest radius. We say that a solution $C'$ is an $\alpha$-approximation, if the cost of $C'$ is at most $\alpha$ times the optimal radius. Several $2$-approximations are known for the $k$-center problem \cite{hochbaumS1985best,hochbaumM1985approximation}. A simple reduction from the Minimum Dominating Set problem shows that the $k$-center problem is \NP-hard. In fact, the same reduction also shows that it is \NP-hard to get a $(2-\epsilon)$-approximation for any $\epsilon > 0$. 

Several generalizations of the vanilla $k$-center problem have been considered in the literature, given its fundamental nature in the domain of clustering and approximation algorithms. One natural generalization is the \emph{Robust $k$-center} or \emph{$k$-center with outliers} problem, where we are additionally given a parameter $m$, and the goal is to find a solution that covers at least $m$ points of $X$. Note that the remaining at most $n-m$ points can be thought of as outliers with respect to the clustering computed. Charikar et al.\ \cite{charikar2001algorithms}, who introduced this problem, showed that a simple greedy algorithm gives a $3$-approximation for the problem. Subsequently, the approximation guarantee was improved by \cite{CGK20,harris2017lottery}, who gave a $2$-approximation, which is optimal in light of the aforementioned $(2-\epsilon)$-hardness result.

The focus of our paper is the Non-Uniform $k$-Center (\nukc), which was introduced by Chakrabarty et al.\ \cite{CGK20}. A formal definition follows.
\begin{definition}[$t$-\nukc] \label{def:nukc}
	The input is an instance $\I = ((X, d), (k_1, k_2, \ldots, k_t), (r_1, r_1, \ldots, r_t))$, where $r_1 \ge r_2 \ge \ldots r_t \geq 0$, and the $k_i$ are positive integers. The goal is to find sets $C_i \subseteq X$ for $1 \le i \le t$, such that $|C_i| \le k_i$, and the union of balls of radius $\alpha r_i$ around the centers in $C_i$, over $1 \le i \le t$, covers the entire set of points $X$. The objective is to minimize the value of the dilation factor $\alpha$.
\end{definition}
In the Robust $t$-\nukc problem, we are required to cover at least $m$ points of $X$ using such a solution. We note that the special case of (Robust) $t$-\nukc with $t = 1$ corresponds to the (Robust) $k$-center problem. Chakrabarty et al.\ \cite{CGK20} gave a bicriteria approximation for $t$-\nukc for arbitrary $t$, i.e., they give a solution containing $O(k_i)$ balls of radius $O(r_i)$ for $1 \le i \le t$. They also give a $(1+\sqrt 5)$-approximation for $2$-\nukc. Furthermore, they conjectured that there exists a polynomial time $O(1)$-approximation for $t$-\nukc for constant $t$. Subsequently, Chakrabarty and Negahbani \cite{ChakrabartyN21} made some progress by giving a $10$-approximation for Robust $2$-\nukc. Very recently, Jia et al.\ \cite{jia2021towards} showed an approximate equivalence between $(t+1)$-\nukc and Robust $t$-\nukc, thereby observing that the previous result of \cite{ChakrabartyN21} readily implies a $23$-approximation for $3$-\nukc. We note that the techniques from Inamdar and Varadarajan \cite{inamdar2020capacitated} implicitly give an $O(1)$-approximation for $t$-\nukc for any $t\ge 1$, in $k^{O(k)} \cdot n^{O(1)}$ time, where $k = \sum_t k_t$. That is, one gets an \emph{FPT approximation}. Finally, we also note that Bandyapadhyay \cite{Bandyapadhyay20nukc} gave an exact algorithm for perturbation resilient instances of \nukc in polynomial time. 

Another related variant of $k$-center is the \emph{Colorful $k$-center} problem. Here, the set of points $X$ is partitioned into $\ell$ \emph{color classes}, $X_1 \cup X_2 \cup \ldots \cup X_\ell$. Each color class $X_j$ has a coverage requirement $m_j$, and the goal is to find a set of $k$ balls of smallest radius that satisfy the coverage requirements of all the color classes. Note that this is a generalization of Robust $k$-center to multiple types of coverage constraints. Bandyapadhyay et al.\ \cite{Bandyapadhyay0P19} introduced this problem, and gave a \emph{pseudo-approximation}, i.e., their algorithm returns an $2$-approximate solution using at most $k+\ell-1$ centers. Furthermore, they managed to improve this to a true $O(1)$-approximation in the Euclidean plane for constant number of color classes. Subsequently, Jia et al.\ \cite{JiaSS20} and Anegg et al.\ \cite{AneggAKZ20} independently gave (true) $3$ and $4$-approximations respectively for the Colorful $k$-center (with constant $\ell$) in arbitrary metrics.

\vspace{0.4cm}\paragraph{Our Results and Techniques.}
Our main result is an $O(1)$-approximation for $4$-\nukc. We obtain this result via a sequence of reductions; some of these reductions are from prior work while some are developed here and constitute our main contribution. Along the way, we combine various tools and techniques from the aforementioned literature of Robust, Colorful, and Non-Uniform versions of $k$-center. 

First, we reduce the $4$-\nukc problem to the Robust $3$-\nukc problem, following Jia et al.\ \cite{jia2021towards}. Next, we reduce the Robust $3$-\nukc to \emph{well-separated} Robust $3$-\nukc, by adapting the approach of Chakrabarty and Negahbani \cite{ChakrabartyN21}.\footnote{In this discussion, ``reduction'' refers to a polynomial time (possibly Turing) reduction from problem $A$ to problem $B$, such that (i) a feasible instance of $A$ yields (possibly polynomially many) instance(s) of $B$, and (ii) a constant approximation for $B$ implies a constant approximation for $A$.} In a well-separated instance, we are given a set of potential centers for the balls of radius $r_1$, such that the distance between any two of these centers is at least $c \cdot r_1$, for a parameter $c \geq 2$.

Before describing how to solve Well-Separated Robust $3$-\nukc, we give a sequence of reductions, which constitute the technical core of our paper. First, we show that any instance of Robust $t$-\nukc can be transformed to an instance of ``Colorful'' $(t-1)$-\nukc, where we want to cover certain number of red and blue points using the specified number of balls of $t-1$ distinct radii. Thus, this reduction reduces the number of radii classes from $t$ to $t-1$ at the expense of increasing the number of coverage constraints from $1$ to $2$. In our next reduction, we show that Colorful $(t-1)$-\nukc can be reduced to Colorful $(t-1)$-\nukc with an additional ``self-coverage'' property, i.e., the radius $r_{t-1}$ can be assumed to be $0$. Just like the aforementioned reduction from \cite{jia2021towards}, these two reductions are generic, and hold for any value of $t \ge 2$. These reductions crucially appeal to the classical greedy algorithm and its analysis from Charikar et al.\ \cite{charikar2001algorithms}, which is a tool that has been not been exploited in the \nukc literature thus far. We believe that these connections between Colorful and Robust versions of \nukc are interesting in their own right, and may be helpful toward obtaining a true $O(1)$-approximation for $t$-\nukc for fixed $t$. 

We apply these two new reductions to transform Well-Separated Robust $3$-\nukc to \emph{Well-Separated Colorful} $2$-\nukc, with $r_{2} = 0$. The latter problem can be solved in  polynomial time using dynamic programming in a straightforward way. Since each of our reductions preserves the approximation factor up to a constant, this implies an $O(1)$-approximation for $4$-\nukc.

Our overall algorithm for $4$-\nukc is combinatorial, except for the step where we reduce Robust $3$-\nukc to Well-Separated Robust $3$-\nukc using the round-or-cut approach of \cite{ChakrabartyN21}. Thus, we avoid an additional ``inner loop'' of round-or-cut that is employed in recent work \cite{ChakrabartyN21,jia2021towards}.\footnote{A by-product of one of our reductions is a purely combinatorial approximation algorithm for colorful $k$-center, in contrast with the LP-based approaches in \cite{Bandyapadhyay0P19,AneggAKZ20,JiaSS20}.}

\section{Definitions, Main Result, and Greedy Clustering}

\subsection{Problem Definitions}

In the following, we set up the basic notation and define the problems we will consider in the paper. We consider a finite metric space $(X, d)$, where $X$ is a finite set of (usually $n$) points, and $d$ is a distance function satisfying triangle inequality. If $Y$ is a subset of $X$, then by slightly abusing the notation, we use $(Y, d)$ to denote the metric space where the distance function $d$ is restricted to the points of $Y$. Let $p \in X$, $Y \subseteq X$, and $r \ge 0$. Then, we use $d(p, Y) \coloneqq \min_{y \in Y} d(p, y)$, and denote by $B(p, r)$ the \emph{ball} of radius $r$ centered at $p$, i.e., $B(p, r) \coloneqq \{ q \in X : d(p, q) \le r \}$. We say that a ball $B(p, r)$ \emph{covers} a point $q$ iff $q \in B(p, r)$; a set of balls $\B$ (resp.\ a tuple of sets of balls $(\B_1, \B_2, \ldots, \B_t)$) covers $q$ if there exists a ball in $\B$ that covers $q$ (resp.\ $\bigcup_{1 \le i \le t} \B_i$ that covers $q$). Analogously, a set of points $Y \subseteq X$ is covered iff every point in $Y$ is covered. For a function $f: S \to \real^+$ or $f: S \to \natural$, and $R \subseteq S$, we define $f(R) \coloneqq \sum_{r \in R} f(r)$.

\begin{definition}[Decision Version of $t$-\nukc]
	\ \\The input is an instance $\I = ((X, d), (k_1, k_2, \ldots, k_t), (r_1, r_2, \ldots, r_t))$, where $r_1 \ge r_2 \ge \ldots r_t \geq 0$, and each $k_i$ is a non-negative integer. The goal is to determine whether there exists a solution $(\B_1, \B_2, \ldots, \B_t)$, where for each $1 \le i \le t$, $\B_i$ is a set with at most $k_i$ balls of radius $r_i$, that covers the entire set of points $X$. Such a solution is called a \emph{feasible solution}, and if the instance $\I$ has a feasible solution, then $\I$ is said to be \emph{feasible}.
	\\An algorithm is said to be an $\alpha$-approximation algorithm (with $\alpha \ge 1$), if given a feasible instance $\I$, it returns a solution $(\B_1, \B_2, \ldots, \B_t)$, where for each $1 \le i \le t$, $\B_i$ is a collection of at most $k_i$ balls of radius $\alpha r_i$, such that the solution covers $X$. 
\end{definition} 
Next, we define the robust version of $t$-\nukc.
\begin{definition}[Decision Version of Robust $t$-\nukc]
	\ \\The input is an instance $\I = ((X, d), (\omega, m), (k_1, k_2, \ldots, k_t), (r_1, r_2, \ldots, r_t))$. The setup is the same as in $t$-\nukc, except for the following: $\omega: X \to \integer^+$ is a weight function, and $1 \le m \le \omega(X)$ is a parameter. The goal is to determine whether there exists a feasible solution, i.e., $(\B_1, \B_2, \ldots, \B_t)$ of appropriate sizes and radii (as defined above), such that the total weight of the points covered is at least $m$. An $\alpha$-approximate solution covers points of weight at least $m$ while using at most $k_i$ balls of radius $\alpha r_i$ for each $1 \le i \le t$.
\end{definition}
We will frequently consider the \emph{unweighted} version of Robust $t$-\nukc, i.e., where the weight of every point in $X$ is unit. Let $\mathbbm{1}$ denote this unit weight function. Now we define the Colorful $t$-\nukc problem, which generalizes Robust $t$-\nukc.
\begin{definition}[Decision Version of Colorful $t$-\nukc]
	\ \\The input is an instance $\I= ((X, d), (\omega_r, \omega_b, m_r, m_b) , (k_1, k_2, \ldots, k_t), (r_1, r_2, \ldots, r_t))$. The setup is similar as in Robust $t$-\nukc, except that we have \emph{two} weight functions $\omega_r, \omega_b: X \to \integer^+$ (corresponding to \emph{red} and \emph{blue} weight respectively). A feasible solution covers a set of points with red weight at least $m_r$, and blue weight at least $m_b$. The notion of approximation is the same as above.
\end{definition}
We note that the preceding definition naturally extends to an arbitrary number $\chi \ge 2$ of colors (i.e., $\chi$ different weight functions over $X$). However, we will not need that level of generality in this paper.

\subsection{Main Algorithm for $4$-\nukc} \label{sec:main-algo}

Let $\I = ((X, d), (k_1, \ldots, k_4), (r_1, \ldots, r_4))$ be the given instance of $4$-\nukc, which we assume is feasible. First, using the reduction Section \ref{sec:t-nukc-to-robust}, we reduce it to an instance $\I' = ((X, d), (\mathbbm{1}, m) (r'_1, r'_2, r'_3), \allowbreak (k_1, k_2, k_3))$ of Robust $3$-\nukc. Recall that Lemma \ref{lem:tplus1-to-t} implies that $\I'$ is feasible, and furthermore an $O(1)$-approximation for $\I'$ implies an $O(1)$-approximation for $\I$.


Next, we use the round-or-cut framework methodology from \cite{ChakrabartyN21} on the instance $\I'$, as described in Section \ref{sec:t-to-robust}. Essentially, this is a Turing reduction from Robust $3$-\nukc to (polynomially many instances of) \emph{Well-Separated} Robust $3$-\nukc. In a well-separated instance, we are given a set of potential centers for the balls of radius $r'_1$, such that the distance between any two potential centers is at least $3r'_1$. At a high level, this reduction uses the ellipsoid algorithm, and each iteration of ellipsoid algorithm returns a candidate LP solution such that, (1) it can be rounded to obtain an $O(1)$-approximate solution for $\I'$, or (2) One can obtain polynomially many instances of \emph{well-separated} Robust $3$-\nukc, at least one of which is feasible, or (3) If none of the obtained instances is feasible, then one an obtain a hyperplane separating the LP solution from the integer hull of coverages.

\paragraph{Solving a Well-Separated Instance.} For the sake of simplicity let $\mathcal{J}$ be one of the instances of Well-Separated Robust $3$-\nukc, along with a well-separated set $Y$ that is a candidate set for the centers of balls of radius $r'_1$. Furthermore, let us assume that $\mathcal{J}$ is feasible. First, the reduction in Section \ref{sec:robust-colorful}, given the instance $\mathcal{J}$, produces $O(n)$ instances $\mathcal{J}(\ell)$ of Colorful $2$-\nukc, such that at least one of the instances is feasible. Then, we apply the reduction from Section \ref{sec:self-coverage} on each of these instances to ensure the \emph{self-coverage} property, i.e., we obtain an instance $\mathcal{J}'(\ell)$ of Colorful $2$-\nukc with $r''_1 = c_2 r'_2 + c_3 r'_3$, and $r''_2 = 0$. Finally, assuming that the resulting instance $\mathcal{J}'(\ell)$ is feasible, it is possible to find a feasible solution using dynamic programming, using the algorithm from Section \ref{sec:dp}. This algorithm supposes that the instance is \emph{Well-Separated} w.r.t.\ a smaller separation factor of $2$. We argue in the next paragraph that this property holds in each each of the instances $\mathcal{J}'(\ell)$.

In order to show that the set $Y$ well-separated w.r.t. the new top level radius $r''_1$, we need to show that $3r'_1 \ge 2r''_1$, i.e., $r'_1 \ge c_2 r'_2 + c_3 r'_3 \ge \beta \cdot r'_2$ for some sufficiently large constant $\beta$. This assumption is without loss of generality, since, if two consecutive radii classes are within a $\beta$ factor, it is possible to combine them into a single radius class, at the expense of an $O(\beta)$ factor in the approximation guarantee. 

Assuming the instance $\mathcal{J}$ is feasible, a feasible solution to an instance $\mathcal{J}'(\ell)$ can be mapped back to an $O(1)$-approximate solution to $\mathcal{J}$, and then to $\mathcal{I}$, since each reduction preserves the approximation guarantee up to an $O(1)$ factor.

\begin{theorem} \label{theorem:main}
	There exists a polynomial time $O(1)$-approximation algorithm for $4$-\nukc.
      \end{theorem}

We have overviewed how the various sections of the paper come together in deriving Theorem \ref{theorem:main}. Before proceeding to these sections, we describe a greedy clustering procedure that we need.

\subsection{Greedy Clustering}
Assume we are given (i)  a metric space $(X, d)$, where $X$ is finite, (ii) a radius $r \geq 0$, (iii) an expansion parameter $\gamma \geq 1$, (iv) a subset $Y \subseteq X$ and a weight function $\omega: Y \to \integer^+$. The weight $\omega(y)$ can be thought of as the multiplicity of $y \in Y$, or how many points are co-located at $y$. We describe a greedy clustering procedure, from Charikar et al. \cite{charikar2001algorithms}, that is used to partition the point set $Y$ into clusters, each of which is contained in a ball of radius $\gamma r$. This clustering procedure, together with its properties, is a crucial ingredient of our approach.

\begin{algorithm}
	\caption{\textsc{GreedyClustering}($Y, X, r \ge 0, \gamma \ge 1, \omega: Y \to \integer^+$)} \label{alg:greedyclustering}
	\begin{algorithmic}[1]
		\Statex We require that $Y \subseteq X$
		\State Let $U \gets Y$, $M \gets \emptyset$ 
		\While{$U \neq \emptyset$}
		\State $p  = \arg\max_{q \in X}\omega(U \cap B(q, r))$ \label{line:argmax}		
		\State $C(p) \coloneqq U \cap B(p, \gamma r)$; $wt(p) \coloneqq \omega(C(p))$ 
		\State $U \gets U \setminus C(p)$
		\State $M \gets M \cup \{p\}$ \Comment{{\small\texttt{We will refer to $p$ as a \emph{mega-point} with cluster $C(p)$ of weight $w(p)$}}}
		\EndWhile
		\State \Return $(M, \{C(p)\}_{p \in M}, \{wt(p)\}_{p \in M})$
	\end{algorithmic}
\end{algorithm}

In line~\ref{line:argmax}, we only consider $q \in X$ such that $U \cap B(q,r) \neq \emptyset$. Notice that it is possible that $\omega(U \cap B(q, r)) = 0$ if $\omega(y) = 0$ for each $y \in U$. Furthermore, notice that we do not require that $q \in U$ for it to be an eligible point in line~\ref{line:argmax}.

We summarize some of the key properties of this algorithm in the following observations.

\begin{observation}
	\label{greedy:observe}
	\begin{enumerate}
		\item For any $p \in M$, $C(p) \subseteq B(p, \gamma r)$,
		\item Point $y \in Y$ belongs to the cluster $C(p)$, such that $p$ is the \emph{first} among all $q \in M$ satisfying $d(y, q) \le \gamma r$. 
		\item The sets $\{C(p)\}_{p \in M}$ partition $Y$, which implies that
		\item $\sum_{p \in M} wt(p) = \omega(Y)$, where $\omega(Z) = \sum_{z \in Z} \omega(z)$ for any $Z \subseteq Y$. 
		\item If $p_i$ and $p_j$ are the points added to $M$ in iterations $i \leq j$, then $wt(p_i) \geq wt(p_j)$.  
		\item For any two distinct $p, q \in M$, $d(p, q) > (\gamma -1) r$.
	\end{enumerate}
\end{observation}
\begin{proof}
	The first five properties are immediate from the description of the algorithm. Now, we prove the sixth property. Suppose for contradiction that there exist $p, q \in M$ with $d(p, q) \le (\gamma -1) r$, and without loss of generality, $p$ was added to $M$ before $q$. Then, note that at the end of this iteration, $B(q, r) \cap U = \emptyset$. Therefore, $q$ will subsequently never be a candidate for being added to $M$ in line~\ref{line:argmax}.
\end{proof}

A key property of this greedy clustering, established by Charikar et al.~\cite{charikar2001algorithms}, is that for any $k \geq 1$ balls of radius $r$, the weight of the points in the first $k$ clusters is at least as large as the weight of the points covered by the $k$ balls.

\begin{lemma}
	\label{charikar-lemma}
	Suppose that the parameter $\gamma$ used in Algorithm~\ref{alg:greedyclustering} is at least $3$. Let $\B$ be any collecion of $k \geq 1$ balls of radius $r$, each centered at a point in $X$. Let $M'$ consist of the first $k'$ points of $M$ chosen by the algorithm, where $k' = \min \{k, |M|\}$. We have
	\[ \sum_{p \in M'} wt(p) = \omega \left( \bigcup_{p \in M'} C(p) \right) \geq \omega \left( Y \cap \bigcup_{B \in \B} B \right). \]
\end{lemma}  

The equality follows from the definition of $wt(p)$ and the fact that the clusters partition $Y$, as stated in Obervation~\ref{greedy:observe}.
\section{From Robust $t$-\nukc to Colorful $(t-1)$-\nukc} \label{sec:robust-colorful}


Let $\I = ((X, d), (\omega, m) , (k_1, k_2, \ldots, k_t), (r_1, r_2, \ldots, r_t))$ be an instance of Robust $t$-\nukc. The reduction to Colorful $(t-1)$-NuKC consists of two phases. In the first phase, we use Algorithm~\ref{alg:greedyclustering} to reduce the instance $\I$ to an instance $\I'$ focused on the cluster centers output by the greedy algorithm. A key property of this reduction is that we may set $r_t = 0$ in the instance $\I'$ -- each ball at level $t$ is allowed to cover at most one point. 

In the second phase, we transform $\I'$ to $O(n)$ instances of Colorful $(t-1)$-\nukc. Assuming there exists a feasible solution for $\I'$, at least one of the instances $\I''$ of Colorful $(t-1)$-\nukc has a feasible solution, and any approximate solution to $\I''$ can be used to obtain an approximate solution to $\I'$ (and thus to $\I$).

\paragraph{Phase 1.}

Let $\I = ((X, d), (\omega, m), (k_1, k_2, \ldots, k_t), (r_1, r_2, \ldots, r_t))$ be an instance of Robust $t$-\nukc. We call the algorithm \textsc{GreedyClustering}$(X, X, r_t, 3, \omega)$, and obtain a set of points $M$ with the corresponding clusters $C(p)$ for $p \in M$. The greedy algorithm also returns a weight $wt(p) = \omega(C(p))$ for each $p \in M$. Let us number the points of $M$ as $p_i$, where $i$ is the iteration in which $p_i$ was added to the set $M$ by algorithm \textsc{GreedyClustering}$(X, X, r_t, 3, \omega)$. This gives an ordering $\sigma = \langle p_1, p_2, \ldots, p_{|M|} \rangle $ of the points in $M$. Note that $wt(p_i) \geq wt(p_j)$ for $i \geq j$. 

We define a weight function $\lambda: X \rightarrow \integer^+$. Let $\lambda(p) = wt(p)$ for $p \in M$ and $\lambda(p) = 0$ for $p \in X \setminus M$. Note that for $p \in M$, $\lambda(p) = wt(p) = \omega(C(p))$. Thus, for each $p \in M$, we are moving the weight from points in cluster $C(p)$ to the cluster center $p$. Clearly, $\omega(X) = \lambda(X)$.

The output of Phase 1 is the instance $\I' = ((X, d), (\lambda, m), (k_1, k_2, \ldots, k_t), (r'_1, r'_2, \ldots, r'_{t-1}, 0))$ of $t$-Robust-NuKC, where $r'_i = r_i + 3r_t$. 
Note that in the instance $\I'$, we have $r'_t = 0$, whereas the other radii in $\I$ have been increased by an additive factor of $3 r_t$. The following claim relates instances $\I$ and $\I'$.

\begin{lemma}
	\label{lemma:phase1}
	(a) If instance $\I$ has a feasible solution, then so does the instance $\I'$. (b)  Given a solution $(\B'_i)_{i \in [t]}$ for $\I'$ that uses at most $k_i$ balls of radius $\alpha r'_i$ for every $i \in [t]$, we can obtain a solution $(\B_i)_{i \in [t]}$ for $\I$ that uses at most $k_i$ balls of radius at most $\alpha r'_i + 3r_t \leq \alpha r_i + (3 \alpha + 3) r_t$ for  $1 \le i \le t$.
\end{lemma}

\begin{proof}
	We begin with part (b). For each ball in $B(p,r)$ that is part of the solution $(\B'_i)_{i \in [t]}$, we replace it with the ball $B(p, r + 3r_t)$ to obtain a solution $(\B_i)_{i \in [t]}$ for $\I$. That is, we expand each ball by an additive $3 r_t$. If $B(p,r)$ covers $q \in M$, then $B(p, r + 3r_t)$ covers $C(q)$, and $\lambda(q) = \omega(C(q))$.  Let $M' \subseteq M$ denote the points covered by $(\B'_i)_{i \in [t]}$. The weight of the points covered by $(\B_i)_{i \in [t]}$ is at least 
	\[ \sum_{p \in M'} \omega(C(p)) = \sum_{p \in M'} \lambda(p) \geq m. \] 
	
	We now establish (a). Fix a feasible solution $(\B_i)_{i \in [t]}$ to $\I$ that covers $\omega$-weight at least $m$, where $\B_i$ is a set of at most $k_i$ balls of radius $r_i$, for $i \in [t]$. Let $M_1 \subseteq M$ be the set of points $p$ such that some point in $C(p)$ is covered by a ball in $\B_1, \B_2, \ldots, \B_{t-1}$. 
	
	Now let $M_2 = M \setminus M_1$ be the set of points $p$, such that any point in $C(p)$ is either covered by a ball from $\B_t$, or is an outlier. Let $X_i \coloneqq \bigcup_{p \in M_i} C(p)$ for $i = 1, 2$. Note that $X = X_1 \sqcup X_2$.

	Note that in the sequence $\sigma = \langle p_1, p_2, \ldots, p_{|M|} \rangle$, the points of $M_1$ and $M_2$ may appear in an interleaved fashion. Let $p_{i_1}, p_{i_2}, \ldots, p_{i_{|M_2|}}$ be the subsequence restricted to the points in $M_2$. In the following lemma, we argue that the first $k_t$ points in this subsequence are sufficient to replace the balls in $\B_t$. Let $k' = \min \{|\B_t|, |M_2|\} \leq k_t$.
	
	\begin{lemma} \label{lem:greedy-ordering}
		There exists a subset $M_2^+ \subseteq M_2$ of size at most $k'$ such that $\sum_{p \in M_2^+} wt(p) \ge \omega \left(X_2 \cap \bigcup_{\ball \in \B_t} \ball \right).$
	\end{lemma}
	\begin{proof}

		Let $M_2^+ = \{p_{i_1}, p_{i_2}, \ldots, p_{i_{k'}}\}$. That is, $M_2^+$ consists of the first $k'$ points of $M_2$ picked by the greedy algorithm. Recall that $M_2^+ \subseteq M_2$, and thus for $p_{i_j} \in M_2^+$, it holds that $C(p_{i_j}) \subseteq X_2$.
		
		Now imagine calling the algorithm \textsc{GreedyClustering}$(X_2, X, r_t, 3, \omega)$. Observe that in the iteration $1 \le j \le |M_2|$, this algorithm will select point $p_{i_j}$ (as defined above) in Line 3, and the corresponding cluster and its weight will be $C(p_{i_j})$ and $wt(p_{i_j})$ -- exactly as in the execution of \textsc{GreedyClustering}$(X, X, r_t, 3, \omega)$. That is, the  algorithm \textsc{GreedyClustering}$(X_2, X, r_t, 3, \omega)$ will output $M_2$ and the clusters $C(p)$ for each $p \in M_2$.
		
		Now, $\B_t$ consists of a set of $|\B_t|$ balls of radius $r_t$. The lemma now follows from Lemma~\ref{charikar-lemma} applied to \textsc{GreedyClustering}$(X_2, X, r_t, 3, \omega)$.
			\end{proof}
	
	Using Lemma \ref{lem:greedy-ordering}, we now construct a solution to instance $\I'$. Fix index $1 \leq i \leq t-1$, and $\B'_i$ denote the set of balls obtained by expanding each ball in $\B_i$ by an additive $3 r_t$. Note that each ball in $\B'_i$ has radius $r'_i = r_i + 3 r_t$.  For every point $p \in M_2^+$, we add a ball of radius $0$ around it and let $\B'_t$ be the resulting set of balls. Note that $|\B'_t| = |M_2^+| \leq k' \leq k_t$.

By definition, for each point $p \in M_1$, there is a ball in $(\B_i)_{i \in [t-1]}$ that intersects cluster $C(p)$, whose points are at distance at most $3 r_t$ from $p$. It follows that the balls in $(\B'_i)_{i \in [t-1]}$ cover each point in $M_1$. 
	
	Using Lemma \ref{lem:greedy-ordering}, the coverage of $(\B'_i)_{i \in [t]}$ in instance $\I'$ is at least
	\[ \sum_{p \in M_1} wt(p) + \sum_{p \in M_2^+} wt(p) \geq \omega(X_1) +  \omega\left(X_2 \cap \bigcup_{\ball \in \B_t} \ball \right) \geq m. \]
	
	The final inequality follows because any point covered by solution $(\B_i)_{i \in [t]}$ for $\I$ either belongs to $X_1$ or to $X_2 \cap \bigcup_{\ball \in \B_t} \ball$. Thus, we have shown that $\I'$ has a feasible solution.
	
\end{proof}

\paragraph{Phase 2.}
Now we describe the second phase of the algorithm. We have the instance $\I' = ((X, d), (\lambda, m), (k_1, k_2, \ldots, k_t), (r'_1, r'_2, \ldots, r'_{t-1}, 0))$ of Robust $t$-\nukc that is output by Phase 1. Phase 2 takes $\I'$ as input and generates an instance $\I(\ell)$, for each $0 \leq \ell \leq |X|$, of the Colorful $(t-1)$-\nukc problem. Note that the number of generated instances is $|X| + 1 = O(n)$. If $\I'$ is feasible, at least one of these $|X| + 1$ instances will be feasible.

Let $\sigma = \langle p_1, p_2, \ldots, p_{|X|} \rangle $ be an ordering of the points in $X$ by non-increasing $\lambda$. That is, $\lambda(p_i) \geq \lambda(p_j)$ for $i \leq j$. 

Fix an index $0 \leq \ell \leq |X|$. We now describe the instance $\I(\ell)$ of colorful $(t-1)$-\nukc. Let $R = \{ p_1, p_2, \ldots, p_{\ell}\}$ denote the set of {\em red} points, and $B = \{ p_{\ell +1}, p_{\ell + 2}, \ldots, p_{|X|} \}$ denote the set of {\em blue} points. For each $p \in B$, define its blue weight as $\omega_b(p) \coloneqq \lambda(p)$; for each $p \in R$, define its blue weight as $\omega_b(p) \coloneqq 0$. Define the blue coverage $m_b$ for instance $\I(\ell)$ as $m_b \coloneqq m - \lambda(R)$. We define the red weight function $\omega_r$ in a slightly different manner. For each red point $p \in R$, let its red weight $\omega_r(p) \coloneqq 1$; for each $p \in B$, let red weight $\omega_r(p) \coloneqq  0$. Let $m_r \coloneqq \sum_{p \in R} \omega_r(p) - k_t = |R| - k_t$ denote the red coverage for instance $\I(\ell)$. Note that $\omega_r$ is supported on $R$ and $\omega_b$ on $B$.  Let $\I(\ell)  \coloneqq ((X, d), (\omega_r, \omega_b, m_r, m_b), (k_1, k_2, \ldots, k_{t-1}), (r'_1, r'_2, \ldots, r'_{t-1}) )$ denote the resulting instance of Colorful $(t-1)$-\nukc problem. Recall that a solution to this instance is required to cover red weight that adds up to at least $m_r$, and blue weight that adds up to at least $m_b$. (In instance $\I(\ell)$, the point sets $R$ and $B$, the red and blue weights, and total coverage requirements $m_r$ and $m_b$ all depend on the index $\ell$. This dependence is not made explicit in the notation, so as to keep it simple.)

We now relate the instance $\I'$ to the instances $\I(\ell)$, for $0 \leq \ell \leq |X|$.

\begin{lemma}
	\label{lemma:phase2}
	(a)   If the instance $\I' = ((X, d), (\lambda, m), (k_1, k_2, \ldots, k_t), (r'_1, r'_2, \ldots, r'_{t-1}, 0) )$ is feasible, then there exists an $0 \leq \ell^* \leq |X|$ such that instance $\I(\ell^*)$ is feasible. \\
	(b) Let $\I(\ell) = ((X, d), (\omega_r, \omega_b, m_r, m_b), (k_1, k_2, \ldots, k_{t-1}), (r'_1, r'_2, \ldots, r'_{t-1}) )$ be a generated instance of Colorful $(t-1)$-NuKc, and suppose $(\B''_i)_{i \in [t-1]}$ is a solution to this instance such that $\B''_i$ contains at most $k_i$ balls of radius $\alpha r'_i$ for $1 \le i \le t-1$, and covers red weight at least $m_r$ and blue weight at least $m_b$. Then, we can efficiently obtain a solution to the instance $\I'$ that uses at most $k_i$ balls of radius $\alpha r'_i$ for $1 \le i \le t -1$, and at most $k_t$ balls of radius $0$.
\end{lemma}

\begin{proof}
	
	We first show part (b). In instance $\I(\ell)$, the red weight $\omega_r(p) = 1$ for each $p \in R$, so the solution $(\B''_i)_{i \in [t-1]}$ covers at least $m_r = \sum_{p \in R} \omega_r(p) - k_t = |R| - k_t$ red points. So the number of red points that are not covered is at most $k_t$. Construct $\B'_t$ by adding a ball of radius $0$ at each uncovered point in $R$. Thus, $|\B'_t| \leq k_t$.  
	
	Let $\B'_i = \B''_i$ for each $1 \leq i \leq t-1$. Now, we argue that the solution $(\B'_i)_{i \in [t]}$ covers weight at least $m$ in instance $\I'$. Note that this solution covers all points in $R$, and a subset $C \subseteq B$ such that $\omega_b(C) \geq m_b = m - \lambda(R)$. Thus the coverage for $\I'$ is at least
	\[ \lambda(R) + \lambda(C) = \lambda(R) + \omega_b(C) \geq \lambda(R) + m_b = m. \] 
	
	We now turn to part (a). Fix a feasible solution $(\B'_i)_{i \in [t]}$ to $\I'$. Let $M_1 \subseteq X$ denote the subset consisting of each point covered by a ball in $\B'_i$, for $1 \leq i \leq t-1$. Let $M_2 = X \setminus M_1$. Each point in $M_2$ is either an outlier or is covered by a ball in $\B'_t$. 
	Note that in the sequence $\sigma = \langle p_1, p_2, \ldots, p_{|X|} \rangle$, the points of $M_1$ and $M_2$ may appear in an interleaved fashion. Let $p_{i_1}, p_{i_2}, \ldots, p_{i_{|M_2|}}$ be the subsequence restricted to the points in $M_2$. Let $k' = \min \{k_t, |M_2|\}$, and let $M_2^+ = \{ p_{i_1}, p_{i_2}, \ldots, p_{i_{k'}} \}$. A key observation is that $\lambda(M_2^+)$ is at least as large as the total weight of the points in $M_2$ covered by balls in $\B'_t$. This is because each ball in $\B'_t$ has radius $0$ and can cover only one point in $M_2$; and the maximum coverage using such balls is obtained by placing them at the points in $M_2$ with the highest weights, i.e, $M_2^+$. Without loss of generality, we assume that $\B'_t$ consists of balls of radius $0$ placed at each point in $M_2^+$.   
	
	Now, let the index $\ell^* \coloneqq i_{k'}$. We now argue that the instance $\I(\ell^*)$ of colorful $(t-1)$-\nukc is feasible. In particular, we argue that $(\B'_i)_{i \in [t-1]}$ is a solution. Consider the set $R = \{p_1, p_2, \ldots, p_{\ell^*}\}$ of red points in $\I(\ell^*)$. Each point in $R$ is either in $M_1$ or in $M_2^+$, and is therefore covered by $(\B'_i)_{i \in [t]}$. It follows that $(\B'_i)_{i \in [t-1]}$ covers at least $|R| - |\B'_t| \geq |R| - |k_t| = m_r$ points of $R$. In other words, the red weight in $\I(\ell^*)$ covered by $(\B'_i)_{i \in [t-1]}$ is at least $m_r$.
	
	Now consider the set $B = \{p_{\ell +1}, p_{\ell + 2}, \ldots, p_{|X|}\}$ of blue points in $\I(\ell^*)$. Let $C \subseteq B$ denote the blue points covered by solution $(\B'_i)_{i \in [t]}$. As $(\B'_i)_{i \in [t]}$ covers points with weight at least $m$ in instance $\I'$, we have $\lambda(R) + \lambda(C) \geq m$; thus, $\lambda(C) \geq m - \lambda(R) = m_b$. However, the balls in $\B'_t$ do not cover any point in $B$. We conclude that the balls in $(\B'_i)_{i \in [t-1]}$ cover all points in $C$. For any $p \in B$, we have $\lambda(p) = \omega_b(p)$. It follows that the blue weight in $\I(\ell^*)$ covered by  $(\B'_i)_{i \in [t-1]}$ is at least $\omega_b(C) = \lambda(C) \geq m_b$. This concludes the proof of part (a).
	
\end{proof}

Combining Lemmas \ref{lemma:phase1} and \ref{lemma:phase2} from Phases 1 and 2, we obtain the following reduction from robust $t$-NuKC to colorful $(t-1)$-NuKC. 

\begin{theorem} \label{thm:t-to-t-1}
	There is a polynomial-time algorithm that, given an instance $\I = ((X, d), (\omega, m), \allowbreak (k_1, \ldots, k_t), (r_1, \ldots, r_t) )$ of Robust $t$-\nukc, outputs a collection of $O(n)$ instances of Colorful $(t-1)$-\nukc with the following properties: (a) If $\I$ is feasible, then at least one of the instances $\I(\ell) = ((X,d), \allowbreak (\omega_r, \omega_b, m_r, m_b), (k_1, \ldots, k_{t-1}), (r'_1, \ldots, r'_{t-1}) )$ of Colorful $(t-1)$-\nukc is feasible; (b) given an $\alpha$-approximate solution to some instance $\I(\ell)$, we can efficiently construct a solution to $\I$ that uses at most $k_i$ balls of radius at most $\alpha r_i + (3\alpha + 3)r_t$. 
\end{theorem}

\begin{remark} In part (a), the  feasible solution for $\I(\ell)$ that is constructed from the feasible solution for $\I$ has the following useful property: for any $\ball$ of radius $r'_i = r_i + 3 r_t$ in the feasible solution for $\I(\ell)$, the center of $\ball$ is also the center of some ball of radius $r_i$ in the feasible solution for $\I$.
\end{remark}
\section{Ensuring Self-Coverage in Colorful $2$-\nukc} \label{sec:self-coverage}

We assume that we are given as input a Colorful $2$-\nukc instance $\I = ((X, d), (\omega_r, \omega_b, m_r, m_b), \allowbreak (k_1, k_2), (r_1, r_2) )$. Recall that $\omega_r: X \rightarrow \integer^+$ (resp. $\omega_b: X \rightarrow \integer^+$) is the red (resp. blue) weight function. The task in Colorful $2$-\nukc is to find a solution $(\B_1, \B_2)$ such that (1) $|\B_i| \le k_i$ for $i = 1, 2$, and (2) the point set $Y \subseteq X$ covered by the solution satisfies $\omega_r(Y) \geq m_r$ and $\omega_b(Y) \geq m_b$, (i.e., the solution covers points with total red weight at least $m_r$, and blue weight at least $m_b$.)  In this section, we show that $\I$ can be reduced to an instance of Colorful $2$-\nukc with $r_2 = 0$. The fact that each ball of radius $r_2$ can only cover its center in the target instance is what we mean by the term {\em self-coverage}. This reduction actually generalizes to Colorful $t$-\nukc, but we address the case $t = 2$ to keep the notation simpler.

Our reduction proceeds in two phases. In Phase 1, we construct an intermediate instance where we can ensure blue self-coverage. Then in Phase 2, we modify the intermediate instance so as to obtain red self-coverage as well.

\paragraph{Phase 1.} In this step, we call the greedy clustering algorithm using the {\em blue weight function} $\omega_b$. In particular, we call \textsc{GreedyClustering}($X, X, r_2, 3, \omega_b$) (See Algorithm \ref{alg:greedyclustering}). This algorithm returns a set of points $M \subseteq X$, where every $p \in M$ has a cluster $C(p)$ and weight $wt(p)$ such that (1) $\{C(p)\}_{p \in M}$ is a partition of $X$; (2) for any $p \in M$, $wt(p) = \omega_b(C(p))$, the blue weight of the cluster, and (3) $d(q, p) \le 3r_2$ for any $q \in C(p)$. Furthermore, the greedy algorithm naturally defines an ordering $\sigma = \langle p_1, p_2, \ldots, p_{|M|} \rangle$ of $M$ -- this is the order in which the points were added to $M$. 

We define a new weight function $\lambda_b: X \rightarrow \integer^+$ as follows: $\lambda_b(p) \coloneqq wt(p)$ if $p \in M$ and $\lambda_b(p) \coloneqq 0$ if $p \in X \setminus M$. Note that for $p \in M$, we have $wt(p) = \omega_b(C(p))$. So the new weight function $\lambda_b$ is obtained from $\omega_b$ by moving weight from each cluster $C(p)$ to the cluster center $p$. 

Phase 1 outputs the intermediate instance $\I' = ((X, d), (\omega_r, \lambda_b, m_r, m_b), (k_1, k_2), (r'_1, r'_2) )$ of Colorful $2$-\nukc, where $r'_1 = r_1 + 6 r_2$ and $r'_2 = 5 r_2$. A solution $(\B'_1, \B'_2)$ for $\I'$ is said to be {\em structured} if it has the following properties.

\begin{enumerate}
\item It is a solution to $\I'$ viewed as an instance of Colorful $2$-\nukc.
\item Let $Y \subseteq X$, the set of points {\em self-covered} by solution $(\B'_1, \B'_2)$, consist of points $p \in X$ such that either (a) $p$ is covered by $\B'_1$, or (b) $p$ is the center of some ball in $\B'_2$.  We require that
\[ \lambda_b(Y) \geq m_b. \]
\end{enumerate}

Thus, a structured solution covers red weight in the usual way; for blue weight, a ball in $\B'_2$ can only contribute blue coverage for its center,

The following lemma relates instances $\I$ and $\I'$. 

\begin{lemma}
	\label{lemma:phase1:colorful}
	(a) If instance $\I$ has a feasible solution, then the instance $\I'$ has a feasible solution that is also structured. (b)  Given a solution $(\B'_1, \B'_2)$ for $\I'$ that uses at most $k_i$ balls of radius $\alpha r'_i$ for every $i  \in \{1, 2\}$, we can obtain a solution $(\B_1, \B_2)$ for $\I$ that uses at most $k_i$ balls of radius $\alpha r'_i + 3r_2 \leq \alpha r_i + (6 \alpha + 3) r_2$ for  $i \in \{1,2\}$.
\end{lemma}

Part (b) is straightforward as the red weights are unchanged in going from $\I$ to $\I'$, and the blue weights are moved by at most $3 r_2$. (Note that we don't require in part (b) that the solution to $\I'$ be structured.)

In the rest of this section, we establish (a). Fix a feasible solution $(\B_1, \B_2)$ to $\I$. Thus, (1) $|\B_i| \le k_i$ for $i = 1, 2$, and (2) the point set $Y \subseteq X$ covered by the solution satisfies $\omega_r(Y) \geq m_r$ and $\omega_b(Y) \geq m_b$, (i.e., the solution covers points with total red weight at least $m_r$, and blue weight at least $m_b$.)

	Let $M_1 \subseteq M$ be the set of points $p$ such that some point in $C(p)$ is covered by a ball in $\B_1$.  Now let $M_2 = M \setminus M_1$ be the set of points $p$ such that any point in $C(p)$ is either covered by a ball from $\B_2$, or is an outlier. Let $X_i \coloneqq \bigcup_{p \in M_i} C(p)$ for $i = 1, 2$. Note that $X = X_1 \sqcup X_2$.
	
	We construct a solution $(\B'_1, \B'_2)$ for instance $\I'$ as follows. The set $\B'_1$ is obtained by expanding each ball in $\B_1$ by an additive factor of $6 r_2$. Thus, the balls in $\B'_1$ cover $X_1$. As in the proof of Lemma~\ref{lemma:phase1}, we construct a subset $N \subseteq M_2$ of size at most $|\B_2|$. We let $\B'_2$ consist of the balls of radius $r'_2 = 5 r_2$, each centered at a point in $N$. The set $N$ will have the following properties:
	
\begin{eqnarray}
\omega_r( X_2 \cap \bigcup_{\ball \in \B'_2} \ball ) & \geq & \omega_r(X_2 \cap \bigcup_{\ball \in \B_2} \ball ) \label{eq:red} \\
\sum_{p \in N} wt(p) & \geq & \omega_b(X_2 \cap \bigcup_{\ball \in \B_2} \ball) \label{eq:blue}
 \end{eqnarray}
 
 It is easy to verify that these two guarantees imply that $(\B'_1, \B'_2)$ is a structured, feasible solution to $\I'$:
 
 The red weight covered by $(\B'_1,\B'_2)$ is at least
 \[ \omega_r(X_1) + \omega_r( X_2 \cap \bigcup_{\ball \in \B'_2} \ball ) \geq \omega_r(X_1) + \omega_r( X_2 \cap \bigcup_{\ball \in \B_2} \ball ) \geq m_r.\]
 
 The set $M_1 \cup N$ is self-covered by $(\B'_1, \B'_2)$. We have
 \[ \lambda_b(M_1)  + \lambda_b(N) = \omega_b(X_1) + \sum_{p \in N} wt(p) \geq \omega_b(X_1) + \omega_b(X_2 \cap \bigcup_{\ball \in \B_2} \ball) \geq m_b. \]

 We now describe the construction of $N$ and establish properties (\ref{eq:red}) and (\ref{eq:blue}). At a high level, this is similar to what we did for $M_2^+$ in Lemma~\ref{lemma:phase1}; but it is more involved as we need to ensure that both properties hold.
 
\begin{algorithm}
	\caption{\textsc{Mapping Procedure}$(\widehat{M}, \sigma, \bhat, \{C(p)\}_{p \in \widehat{M}})$} \label{alg:mapping}
	\begin{algorithmic}[1]
		\State Index the points of $\widehat{M}$ as $q_1, q_2, \ldots$ according to the ordering $\sigma$
		\State For every $\ball \in \bhat$, $\varphi(\ball) \coloneqq q_i$, where $q_i \in \widehat{M}$ is the \emph{first} point $q$ s.t. $\ball \cap C(q) \neq \emptyset$
		\State $\ell = 0$; $\mathcal{T} \gets \emptyset$
		\While{there exists a $\ball \in \bhat$ that does not belong to any $D_{j}$ with $j \le \ell$}
		\State $\ell \gets \ell+1$
		\State $q_i \in \widehat{M} \setminus \bigcup_{j = 1}^{\ell-1} N_\ell$ be the first point $q$ with $|\varphi^{-1}(q)| > 0$
		\State \texttt{pending} $ \gets |\varphi^{-1}(q_i)| - 1$
		\State $N_\ell \gets \{q_i\}$, $D_\ell \gets \varphi^{-1}(q_i)$
		\While{$\texttt{pending} > 0$ \textbf{ and } $i+1 \le |\widehat{M}|$}
		\State $i \gets i+1$
		\State $\texttt{pending} \gets \texttt{pending} + |\varphi^{-1}(q_i)|  - 1$
		\State $N_\ell \gets N_\ell \cup \{q_i\}$, $D_\ell \gets D_\ell \cup \varphi^{-1}(q_i)$
		\EndWhile
		\State Add $(N_\ell, D_\ell)$ to $\mathcal{T}$
		\EndWhile
		\State Return $\mathcal{T}$
	\end{algorithmic}
\end{algorithm}

Let $\bhat_2 = \{\ball \in \B_2 \ | \ \ball \cap X_2 \neq \emptyset \}$. 
The set $N$ is obtained via \textsc{Mapping Procedure}, given in Algorithm \ref{alg:mapping}. In particular, we invoke \textsc{Mapping Procedure}($M_2, \sigma, \bhat_2, \{C(p\}_{p \in M_2}$). We describe Algorithm \ref{alg:mapping} at a high level. First, we map every ball in $\bhat_2$ to the \emph{first} (according to $\sigma$) point $q$ in $M_2$ whose cluster $C(q)$ has a non-empty intersection with the ball -- this is the definition of $\varphi$. Now, some points $q \in M_2$ may get mapped by more than one ball. Then, we create a ``grouping procedure'' that creates pairs $(N_\ell, D_\ell)$ as follows. We start from the \emph{first} (according to $\sigma$) point $q_i$ that is mapped by at least one ball. We add $q_i$ to $N_\ell$, and the balls that were mapped to $q_i$ to the set $D_\ell$. Now, if $|\varphi^{-1}(q_i)| > 1$, then we aim to find $|\varphi^{-1}(q_i)| -1$ additional points after $q_i$ to be added to $N_\ell$. Furthermore, it is important in the analysis that these points be \emph{consecutive} according to $\sigma_{|M_2}$. The variable \texttt{pending} keeps track of how many additional distinct points need to be added to $N_\ell$ to match the number of distinct balls in $D_\ell$ at the current time. Thus, if $|\varphi^{-1}(q_i)| > 1$, we add $q_{i+1}$ to $N_\ell$ as well. At this stage, it may happen that $\varphi^{-1}(q_{i+1}) \neq \emptyset$. Then, we add $\varphi^{-1}(q_{i+1})$ to $D_\ell$, and update the variable \texttt{pending} appropriately. If the variable \texttt{pending} becomes $0$, then $|N_\ell| = |D_\ell|$, at which point the inner while loop terminates. By construction, the points added to $N_\ell$ form a contiguous sub-sequence of $\sigma_{|M_2}$. We add the pair $(N_\ell, D_\ell)$ to $\T$. At this point, if there still exists a ball of $\bhat_2$ that does not belong to any $D_j$ with $j \le \ell$, we start the construction of the next pair $(N_{\ell+1}, D_{\ell+1})$. Note that in all but the last iteration of the outer while loop, it holds that $|N_\ell| = |D_\ell|$. However, in the last iteration $t$, the loop may terminate with $|N_t| \le |D_t|$. 

The invocation of \textsc{Mapping Procedure}($M_2, \sigma, \bhat_2, \{C(p\}_{p \in M_2}$) returns 
$\mathcal{T} = \{(N_1, D_1),  (N_2, D_2), \allowbreak \ldots, (N_t, D_t)\}$.  In the following observation, we summarize a few key properties of this collection of pairs.

\begin{observation} \label{obs:mapping-properties}
	$\mathcal{T} = \{(N_1, D_1), (N_2, D_2), \ldots, (N_t, D_t)\}$ satisfies the following properties.
	\begin{enumerate}
		\item For each $1 \leq \ell \leq t$, we have $\emptyset \neq N_\ell \subseteq M_2$; Furthermore, the points of $N_\ell$ form a contiguous subsequence of $M_2$ ordered according to $\sigma$.
The sets $N_1, N_2, \ldots, N_t $ are pairwise disjoint.
		\item For each $1 \leq \ell \leq t$, we have $\emptyset \neq D_\ell \subseteq \bhat_2$. The sets $D_1, D_2, \ldots, D_t$ form a partition of $\bhat_2$.
		\item $|N_\ell| = |D_\ell|$ for $\ell < t$, and $|N_t| \le |D_t|$.
	\end{enumerate}
\end{observation}

Now we prove the following key lemma.
\begin{lemma} \label{lem:blue-greedy}
	For any $1 \leq \ell \leq t$, the following properties hold.
	\begin{enumerate}[label=(\Alph*)]
		\item For any ball $B(c, r_2) \in D_\ell$, there exists a $q \in N_\ell$ such that $B(c, r_2) \subseteq B(q, 5r_2)$.  \label{prop:one}
		\item $\displaystyle \omega_b\lr{ X_2 \cap \bigcup_{B(c, r_2) \in D_\ell} B(c, r_2)} \le \sum_{p \in N_\ell} wt(p) $. \label{prop:two}
	\end{enumerate}
\end{lemma}
\begin{proof}
	
	For any $\ball = B(c, r_2) \in D_\ell$, $q_i = \varphi(\ball) \in N_\ell$. By the definition of $q_i$, it holds that $C(q_i) \cap \ball \neq \emptyset$. Therefore, for any point $p \in \ball$, it holds that $d(p, q_i) \le d(p, c) + d(c, p') + d(p', q_i) \le r_2 + r_2 + 3r_2 = 5r_2$, where $p' \in C(q_i) \cap \ball$. This proves property \ref{prop:one}.
	
	Let $\mathcal{X}_\ell \coloneqq X_2 \cap \lr{\lr{\bigcup_{q \in N_\ell} C(q)} \cup \lr{ \bigcup_{\ball \in D_\ell} \ball } }$. That is, $\mathcal{X}_\ell$ denotes the set of those points in $X_2$ that belong to the clusters of all the points in $N_\ell$, as well as those in the balls in $D_\ell$. Now, imagine calling \textsc{GreedyClustering}$(\mathcal{X}_\ell, X, r_2, 3, \omega_b)$. As in the proof of Lemma \ref{lem:greedy-ordering}, the main observation is that the set of clusters computed in the first $|N_\ell |$ iterations  is exactly $\{C(q)\}_{q \in N_\ell}$. Thus, property \ref{prop:two} in the lemma follows from 
Lemma~\ref{charikar-lemma} applied to \textsc{GreedyClustering}$(\mathcal{X}_\ell, X, r_2, 3, \omega_b)$.
\end{proof}

We now set $N = \bigcup_{1 \leq \ell \leq t} N_{\ell}$. Note that 
\[ |N| = \sum_{\ell} |N_\ell | \leq \sum_{\ell} |D_\ell | = |\bhat_2| \leq |\B_2|. \]
Recall that for instance $\I'$, we set $\B'_2 = \{ B(q, 5r_2) \ | \ q \in N \}.$ We now argue that $N$ satisfies properties (\ref{eq:red}) and (\ref{eq:blue}).

By Property \ref{prop:one} of Lemma \ref{lem:blue-greedy}, we have that for any $\ball \in \B_2$, there is a $\ball' \in \B'_2$ such that $X_2 \cap \ball \subseteq X_2 \cap \ball'$. Thus, $\lr{X_2 \cap \bigcup_{\ball \in \B_2} \ball} \subseteq \lr{X_2 \cap \bigcup_{\ball \in \B'_2} \ball}$, which implies property (\ref{eq:red}).

Using Property \ref{prop:two} of Lemma \ref{lem:blue-greedy}, we have

\[\sum_{p \in N} wt(p)  = \sum_{\ell} \sum_{p \in N_\ell} wt(p)  \geq  \sum_{\ell} \omega_b \left(X_2 \cap \bigcup_{\ball \in D_\ell} \ball \right)  \geq \omega_b \left(X_2 \cap \bigcup_{\ball \in \B_2} \ball \right),\]
which is property (\ref{eq:red}).

\textbf{Phase 2.}
Phase 1 outputs an instance $\I' = ((X, d), (\lambda_r, \lambda_b, m_r, m_b), (k_1, k_2), (r'_1, r'_2) )$ of Colorful $2$-\nukc. In Phase 2, we transform this into an instance $\I'' = ((X,d), (\chi_r, \chi_b, m_r, m_b), (k_1, k_2), (r''_1, 0) )$ of Colorful $2$-\nukc where the radius at the second level is $0$. 

In this step, we call the greedy clustering algorithm (Algorithm \ref{alg:greedyclustering}) using the {\em red weight function} $\lambda_r$. In particular, we will call \textsc{GreedyClustering}($X, X, r'_2, 3, \lambda_r$). This algorithm returns a set of points $M \subseteq X$, where every $p \in M$ has a cluster $C(p)$ and weight $wt(p)$ such that (1) $\{C(p)\}_{p \in M}$ is a partition of $X$, (2) For any $p \in M$, $wt(p) = \lambda_r(C(p))$, the red weight of the cluster, and (3) $d(q, p) \le 3r'_2$ for any $q \in C(p)$. Furthermore, the greedy algorithm naturally defines an ordering $\sigma = \langle p_1, p_2, \ldots, p_{|M|} \rangle$ of $M$ -- this is the order in which the points were added to $M$.

We define the red weight function $\chi_r$ for $\I''$ as follows: $\chi_r(p) \coloneqq \lambda_r(C(p))$ for $p \in M$, and $\chi_r(p) \coloneqq 0$ for $p \in X \setminus M$. 

We define a  $\phi: X \to M $ as follows: $\phi(p)$ is the first point in $M$ (according to $\sigma$) such that $B(p, r'_2) \cap C(p) \neq \emptyset$. Note that $\phi(p)$ exists and $d(p, \phi(p)) \leq 4 r'_2$. We define the blue weight function $\chi_b$ for $\I''$ as follows:  $\chi_b(p) \coloneqq \sum_{q \in \phi^{-1}(p)} \lambda_b(q)$ for $p \in M$, and $\chi_b(p) \coloneqq 0$ for $p \in X \setminus M$. 

Finally, we let $r''_1 = r'_1 + 4r'_2$, and obtain the instance $\I'' = ((X,d), (\chi_r, \chi_b, m_r, m_b), (k_1, k_2), (r''_1, 0) )$ of Colorful $2$-\nukc. The following lemma relates instances $\I'$ and $\I''$.

\begin{lemma}
	\label{lemma:phase2:colorful}
	(a) If instance $\I'$ has a feasible solution that is structured, then the instance $\I''$ has a feasible solution. (b)  Given a solution $(\B''_1, \B''_2)$ for $\I'$ that uses at most $k_i$ balls of radius $\alpha r''_i$ for each $i  \in \{1, 2\}$, we can obtain a solution $(\B'_1, \B'_2)$ for $\I'$ that uses at most $k_i$ balls of radius $\alpha r''_i + 4r'_2 \leq \alpha r'_i + (4 \alpha + 4) r'_2$ for  $i \in \{1,2\}$.
\end{lemma}

Again, part (b) follows from the fact that in constructing $\I''$ from $\I'$, we move weights by a distance of at most $4 r'_2$. Note that we do not claim that the solution to $\I'$ constructed in part (b) is structured.

In the rest of this section, we establish part (a). Fix a feasible solution $(\B'_1, \B'_2)$ for $\I'$ that is also structured. Our construction of a feasible solution for $\I'$ is analogous to what we did in Phase 1. 

Let $M_1 \subseteq M$ be the set of points $p$ such that there exists some point $x$ satisfying (i) $x$ is covered by a ball in $\B'_1$, and (ii) $d(x,p) \leq 4 r'_1$. Note that $M_1$ includes any $p \in M$ such $C(p)$ contains a point covered by a ball in $\B'_1$. Now let $M_2 = M \setminus M_1$; note that for $p \in M_2$, any point in $C(p)$ is either covered by a ball from $\B'_2$, or is an outlier. Let $X_i \coloneqq \bigcup_{p \in M_i} C(p)$ for $i = 1, 2$. Note that $X = X_1 \sqcup X_2$.

Let $\bhat'_2 = \{ \ball \in \B'_2 \ | \ \ball \cap X_2 \neq \emptyset \}.$ We invoke \textsc{Mapping Procedure}($M_2, \sigma, \bhat'_2, \{C(p\}_{p \in M_2}$) and  $\mathcal{T} = \{(N_1, D_1), (N_2, D_2), \ldots, (N_t, D_t)\}$. We let $N = \bigcup_{1 \leq \ell \leq t} N_{\ell}$.

As in phase 1, we have that $|N| \leq |\bhat'_2| \leq | \B'_2 |$. The set $N$ also satisfies the following property, which is the analog of Property \ref{eq:blue}.

\begin{equation}
\sum_{p \in N} wt(p)  \geq  \lambda_r(X_2 \cap \bigcup_{\ball \in \B'_2} \ball) \label{eq:blue-red}
\end{equation}

We now construct a solution $(\B''_1, \B''_2)$ for $\I''$. The set $\B''_1$ is obtained by expanding each ball in
$\B'_1$ by an additive $4 r'_2$; each ball in $\B''_1$ has radius $r''_1$. Note that by definition of $M_1$, the balls in $\B''_1$ cover $M_1$. The set $\B''_2$ is obtained by including in it a ball of radius $0$ at each point in $N$. Note that $|\B''_2| = |N|  \leq |\B'_2|$. 

We now argue that $(\B''_1, \B''_2)$ provides adequate coverage. Red coverage is analogous to blue coverage in phase 1, using property \ref{eq:blue-red}:
\[ \chi_r(M_1)  + \chi_r(N) = \lambda_r(X_1) + \sum_{p \in N} wt(p) \geq \lambda_r(X_1) + \lambda_r(X_2 \cap \bigcup_{\ball \in \B'_2} \ball) \geq m_r. \]

For blue coverage, let $Y \subseteq X$ denote the set of points {\em self-covered} by the structured, feasible solution $(\B'_1, \B'_2)$ with $\lambda_b(Y) \geq m_b$. We argue that for each $y \in Y$, we have $\phi(y) \in M_1 \cup N$. If $y$ is covered by a ball in $\B'_1$, then as $d(y, \phi(y)) \leq 4 r'_2$, we conclude that $\phi(y) \in M_1$ using the definition of $M_1$. Otherwise, $y$ is the center of some ball in $B(y,r'_2) \in \B'_2$. Assume $\phi(y) \not\in M_1$. Then by the definition of $\phi$, $\phi(y)$ is the first point $p \in M_2$ such that $B(y, r'_2)$ intersects $C(p)$. But this means $\phi(y)$ is the same as $\varphi( B(y, r'_2) )$ computed in \textsc{Mapping Procedure}($M_2, \sigma, \bhat'_2, \{C(p\}_{p \in M_2}$). Thus, $B(y, r'_2) \in D_{\ell}$ and $\phi(y) \in N_{\ell}$ for some pair
$(N_{\ell}, D_{\ell})$ in $\mathcal{T}$. We conclude $\phi(y) \in N = \bigcup_{\ell} N_\ell$.

Thus, the blue coverage of $(\B''_1, \B''_2)$ is at least
\[\chi_b(M_1) + \chi_b(N) \geq \sum_{p \in M_1 \cup N}  \phi^{-1}(p) \geq \sum_{y \in Y} \lambda_b(y) \geq m_b.\]

This completes the proof of Lemma~\ref{lemma:phase2:colorful} and concludes our description of Phase 2. Combining Phase 1 and Phase 2, we conclude with the main result of this section.

\begin{theorem}
\label{theorem:self-coverage}
There is a polynomial-time algorithm that transforms a Colorful $2$-\nukc instance $\I = ((X, d), (\omega_r, \omega_b, m_r, m_b), (k_1, k_2), (r_1, r_2) )$ into an instance $\I''= ((X,d), (\chi_r, \chi_b, m_r, m_b), (k_1, k_2), (r''_1, 0) )$ of Colorful $2$-\nukc with $r''_1 = r_1 + 26 r_2$, and has the following properties: (a) If $\I$ has a feasible solution, then so does $\I''$; (b) Given an $\alpha$-approximate solution to $\I''$, we can construct, in polynomial time, a $c \cdot \alpha$-approximate solution to $\I$, where $c > 0$ is an absolute constant.
\end{theorem}

\begin{remark} In part (a), the feasible solution $(\B''_1, \B''_2)$ to $\I'$ that is constructed from feasible solution $(\B_1, \B_2)$ to $\I$ has the following useful property: for any $\ball \in \B''_1$, the center of $\ball$ is also the center of some ball in $\B'_1$.
\end{remark}

\section{Solving Well-Separated Colorful $2$-NUkC} \label{sec:dp}

We assume that we are given a \emph{well-separated} instance $\I = ((X, d), (\omega_r, \omega_b, m_r, m_b) (k_1, k_2), (r_1, 0))$ of Colorful $2$-\nukc. The \emph{well-separatedness} of the instance comes with the following additional input and restriction -- we are given an additional set $Y \subseteq X$ as an input. The set $Y$ is \emph{well-separated}, i.e., for any $u, v \in Y$, $d(u, v) > 2 r_1$. The additional restriction is that, the set of centers of balls of radius $r_1$ must be chosen from the set $Y$. We sketch how to solve such an instance optimally in polynomial time using dynamic programming.

Let $z \coloneqq |Y|$, and let $Y = \{y_1, y_2, \ldots, y_z\}$. For $1 \le i \le z$, let $X_i \coloneqq B(y_i, r_1) \cap X$, and let $X_{z+1} \coloneqq X \setminus \lr{\bigcup_{1 \le i \le z} X_i}$. Note that $\{X_i\}_{1 \le i \le z+1}$ is a partition of $X$. 

For any $X' \subseteq X$ and non-negative integers $k, n_r, n_b $, let $F(X', k, n_r, n_b)$ be \textbf{true} if there exists a subset $X'' \subseteq X'$ of size at most $k$, and (red, blue) weight at least $(n_r, n_b)$; and \textbf{false} otherwise.\footnote{We use \emph{$X''$ has (red, blue) weight at least $(n_r, n_b)$} as shorthand for $\omega_r(X'') \geq n_r$ and $\omega_b(X'') \geq n_b$.} For a particular subset $X'$, the value of $F(X', k, n_r, n_b)$ can be found in polynomial time using dynamic programming, since the values $k, n_r, n_b$ are at most $n$.

For $(1, 0, 0, 0, 0) \le (i, k'_1, k'_2, n_r, n_b) \le (z+1, k_1, k_2, m_r, m_b)$, let $G(i, k'_1, k'_2, n_r, n_b)$ be \textbf{true} if it is possible to obtain (red, blue) coverage of at least $(n_r, n_b)$ from the set of points $\bigcup_{1 \le j \le i} X_j$, using at most $k'_1$ balls of radius $r_1$ and $k'_2$ balls of radius $0$;  and \textbf{false} otherwise. Note that if $G(i-1, k'_1, k'_2, n_r, n_b) = \textbf{true}$, then $G(i, k'_1, k'_2, n_r, n_b)$ is trivially \textbf{true}. Otherwise, suppose some points in $X_i$ are covered. We consider two possibilities: either (A) $X_i$ is covered using a ball of radius $r_1$ (note that for $i \le z$ this is possible by definition; for $i = z+1$ we omit this case), and the remaining (red, blue) coverage comes from $\bigcup_{1 \le j \le i-1} X_j$, or (B) We use some $1 \le t \le \min\{k'_2, |X_i|\}$ balls of radius $0$ to achieve the (red, blue) coverage of $(n'_r, n'_b)$ from within $X_i$, and the remaining (red, blue) coverage comes from $\bigcup_{1 \le j \le i-1} X_j$. Note that in case (B), for a fixed guess of $(t, n'_r, n'_b)$, the subproblem for $X_i$ corresponds to $F(X_i, t, n'_r, n'_b)$ as defined in the previous paragraph, and can be solved in polynomial time. It is straightforward to convert this recursive argument to compute $G(z+1, k_1, k_2, m_r, m_b)$ into a dynamic programming algorithm that also finds a feasible solution, and it can be implemented in polynomial time. We omit the details.

\section{From Robust $t$-\nukc to Well-Separated Robust $t$-\nukc} \label{sec:t-to-robust}

In this section, we use the round-or-cut framework of \cite{ChakrabartyN21} to give a Turing reduction from Robust $t$-\nukc to (polynomially many instances of) \emph{Well-Separated} Robust $t$-\nukc. Furthermore, $c$-approximation for a feasible instance of the latter problem will imply an $O(c)$-approximation for the original instance of Robust $t$-\nukc. 

\paragraph{Round-or-Cut Framework.} Let $\I= ((X, d), (\mathbbm{1}, m), (k_1, k_2, \ldots, k_t), (r_1, r_2, \ldots, r_t))$ be the given instance of Robust $t$-\nukc (we assume that we are working with unit-weight instance, where we want to cover at least $m$ points of $X$). We adopt the round-or-cut framework of \cite{ChakrabartyN21} (also \cite{chakrabarty2018generalized}) to separate an LP solution from the integer hull of coverages (see Section \ref{sec:setup} in the appendix for the definitions thereof). Even though \cite{ChakrabartyN21} discuss this for $t = 2$, it easily generalizes to arbitrary $t \ge 2$. Thus, we only sketch the high level idea. 

Let $\cov = (\cov_1, \cov_2, \ldots, \cov_t: \forall v \in X)$ be a candidate solution returned by the ellipsoid algorithm. First, we check whether $\cov(X) \ge m$, and report as the separating hyperplane if this does not hold. Now, we call CGK Algorithm (see Section \ref{sec:setup}) with $\alpha_1 = 6$, and $\alpha_i = 2$ for all $2 \le i \le t$ to get a $t$-FF instance $(\T = ((L_1, \ldots, L_t), (a_1, \ldots, a_t), \mathsf{Leaf}, w), (k_1, \ldots, k_t))$. Here, for any $i \in [t]$, any distinct $p, q \in L_i$ satisfy that $d(p, q) > 3r_i$. Then, we let $\{y_v: v \in \bigcup_{i} L_i\}$ be the solution as defined in Section \ref{sec:setup}, see Definition \ref{def:soln-y}. Now we check if $\cov_i(L_i) \le k_i$ for $i \in [t]$, and report if any of these $t$ inequalities is not satisfied. Finally, the algorithm checks the value of $y(L_1)$, and branches into the following two cases.

In the first case, if $y(L_1) \le k_1 - t$, then as argued by \cite{ChakrabartyN21}, it can be shown that a sparse LP that is related to the $t$-FF problem (see Definitions \ref{def:t-ff} and \ref{def:sparse-lp}) admits an \emph{almost-integral} solution. That is, a basic feasible solution to the sparse LP contains at most $t$ strictly fractional variables. By rounding up all such variables to $1$, one can obtain an $O(1)$-approximation for the original instance $\I$. Note that here we need the assumption that the ratio between the values of consecutive radii is at least $\beta$ -- otherwise we can merge the two consecutive radii classes into a single class.

In the second case, $y(L_1) > k_1 - t$. In this case, we use a generalization of an argument from \cite{ChakrabartyN21} as follows. We enumerate every subset $Q \subseteq X$ of size at most $t-1$, and add a ball of radius $r_1$ around each point in $Q$. Let $X'$ be the set of points covered by balls of radius $r_1$ around $Q$. Then, we modify the weight of the points of $X'$ to be $0$, and let $\mathbbm{1}_{X\setminus X'}$ be the resulting weight function. Let $\I(Q) = ((X, d), (\mathbbm{1}_{X \setminus X'}, m-|X'|), (2r_1, r_2, \ldots, r_t), (k_1 - |Q|, k_2, \ldots, k_t))$ be the resulting residual instance of \emph{Well-Separated} $t$-\nukc, where the \emph{well-separatedness} property imposes that the $2r_1$ centers must be chosen from $Y \coloneqq L_1 \setminus Q$ -- note that the distance between any two distinct points in $L_1$, and thus $Y$, is at least $6r_1 = 3 \cdot 2r_1$, i.e., the set $Y$ is well-separated w.r.t. the new radius $r_1$. An argument from \cite{ChakrabartyN21} implies that if $\I$ is feasible, then either (a) at least one of the well-separated instances $\I(Q)$ is feasible for some $Q \subseteq X$ of size at most $t-1$, or (b) the hyperplane $y(L_1) \le k_1 -t$ separates the LP solution $\cov$ from the integer hull of coverages. Furthermore, an argument from \cite{ChakrabartyN21} implies that a constant approximation to any of the instances implies a constant approximation to $\mathcal{I}$. 

Note that the ellipsoid algorithm terminates in polynomially many iterations, and each iteration produces at most $n^{t}$ instances of \emph{Well-Separated} Robust $t$-\nukc. Thus, we get the following theorem.

\begin{theorem} \label{theorem:robust-to-ws}
	Suppose there exists an algorithm that, given an instance $\mathcal{J}$ of \emph{Well-Separated} Robust $t$-\nukc, in time $f(n, t)$, either finds an $\alpha$-approximation to $\mathcal{J}$, or correctly determines that $\mathcal{J}$ is not feasible. Then, there exists an algorithm to obtain an $c \cdot \alpha$-approximation for any instance of Robust $t$-\nukc, running in time $n^{O(t)} \cdot f(n, t)$. 
\end{theorem}

\bibliography{references}

\appendix

\section{From $(t+1)$-NUkC to Robust $t$-NUkC} \label{sec:t-nukc-to-robust}

In this section, we show an approximate equivalence of $t+1$-NUkC and Robust $t$-NUkC. Note that Jia et al.\ \cite{jia2021towards} recently showed a very similar result. However, our proof is slightly different from theirs, and we describe it here for the sake of completeness.

\begin{lemma}\label{lem:tplus1-to-t}\ 
	\begin{enumerate}
		\item Suppose there exists an $\alpha$-approximation algorithm for $(t+1)$-NUkC. Then, there exists an $\alpha$-approximation algorithm for unweighted Robust $t$-NUkC.
		\item Suppose there exists a $\beta$-approximation algorithm for unweighted Robust $t$-NUkC. Then there exists a $3\beta+2$-approximation algorithm for $(t+1)$-NUkC.
	\end{enumerate}
\end{lemma}
\begin{proof}
	Note that the first claim is trivial, since an instance of Robust $t$-NUkC is a special case of NUkC, as follows. Let $\I = ((X, d), (\mathbbm{1}, m), (r_1, r_2, \ldots, r_t), (k_1, k_2, \ldots, k_t))$ be an instance of unweighted $t$-Robust-NUkC, where $m$ is the coverage requirement. Then, observe that it is equivalent to the instance $\I' = ((X, d), (r_1, r_2, \ldots, r_t, 0), (k_1, k_2, \ldots, k_t, n-m))$ of $t+1$-NUkC. An $\alpha$-approximate solution to $\I'$ immediately gives an $\alpha$-approximate solution to $\I$. We now proceed to the second claim.
	
	Consider an instance $\I = ((X, d), (r_1, r_2, \ldots, r_t, r_{t+1}), (k_1, k_2, \ldots, k_t, k_{t+1}))$ of $(t+1)$-NUkC. Note that we have to cover all points of $X$ in the instance $\I$. First, we compute a $2r_{t+1}$-net $Y$ of $X$. That is compute $Y \subseteq X$ with the following properties: (i) $d(u, v) > 2r_{t+1}$ for any $u, v \in Y$, and (ii) for any $u \in X \setminus Y$, there exists a $v \in Y$ such that $d(u, v) \le 2r_{t+1}$. Let $\varphi: X \to Y$ be a mapping that assigns every point in $X$ to its nearest point in $Y$ (breaking ties arbitrarily). Our reduction constructs the instance $\I' = ((Y, d), (\mathbbm{1}, |Y|-k_{t+1}), (k_1, k_2, \ldots, k_t), (r'_1, r'_2, \ldots, r'_t))$ of $t$-Robust-NUkC with at most $k_{t+1}$ outliers, where $r'_i = r_i + 2r_t$ for $1 \le i \le t$.
	
	We now argue that if $\I$ is feasible, then so is $\I'$.
	Fix a solution $(\B_i)_{i \in [t+1]}$ for the original instance $\I$, where $\B_i$ is a set of at most $k_i$ balls of radius $r_i$. Let $Y' \subseteq Y$ be the set of points in $Y$ covered by $(\B_i)_{i \in [t]}$, the balls of the $t$ largest radii types. For each ball $B(c_i, r_i) \in \B_i$, we add $B(\varphi(c_i), r'_i)$ to obtain the set $\B'_i$ of balls; recall $r'_i = r_i + 2r_t$. Note that the resulting solution $(\B'_i)_{i \in [t]}$ covers the set of points $Y'$. Now, let $Y'' = Y \setminus Y'$ be the set of points covered by $\B_{t+1}$, the balls of radius $r_{t+1}$. The distance between any two points of $Y$, and thus $Y''$, is greater than $2r_{t+1}$. Therefore, a ball of radius radius $r_{t+1}$ covers at most one point of $Y''$, which implies that $|Y''| \le |\B_{t+1}| \le k_{t+1}$. Thus $(\B'_i)_{i \in [t]}$ is a feasible solution for instance $\I'$, with the points in $Y''$ being the set of outliers of size at most $k_{t+1}$.
	
	We now argue that from a $\beta$-approximate solution to $\I'$, we can efficiently construct a $(3\beta + 2)$-approximate solution to $\I$. Fix a solution $(\B'_i)_{i \in [t]}$ for the instance $\I'$ that covers at least $|Y|-k_{t+1}$ points of $Y$, where $\B'_i$ consists of $k_i$ balls of radius $\beta r'_{i}$, for $1 \le i \le t$. To obtain a solution for the original instance $\I$, we proceed as follows. We expand the radius of every ball in $\B'_i$ by an additive factor of $2r_{t+1}$ to obtain $\B_i$. Note that the resulting radius for each ball in $\B_i$ is $\beta r_i + 2\beta r_{t+1} + 2r_{t+1} \le (3\beta+2) \cdot r_{i}$. Note that if a ball in solution $(\B'_i)_{i \in [t]}$ covers $y \in Y$, then the additively expanded version of the ball covers every point $x \in \varphi^{-1}(y)$. For every outlier point $y \in Y$ not covered by $(\B'_i)_{i \in [t]}$, we add a ball of radius $2r_{t+1}$ centered at $y$ to $\B_{t+1}$; this ball covers all points $x \in \varphi^{-1}(y)$. As the number of outliers is at most $k_{t+1}$, we have $|\B_{t+1}| \leq k_{t+1}$. The resulting solution $(\B_i)_{i \in [t+1]}$ covers all the points of $X$, and has approximation guarantee $3\beta+2$.
\end{proof}

\section{Setup for Robust $t$-\nukc} \label{sec:setup}
Let $\I = ((X, d), (\mathbbm{1}, m) (k_1, \ldots, k_t), (r_1, \ldots, r_t))$ be an instance of Robust $t$-NUkC. First we state the natural LP relaxation for $\I$. Recall that the goal is to cover at least $m$ points.

\begin{align*}
	\sum_{v \in X} \cov(v) &\ge m 
	\\\sum_{u \in X} x_{i, u} &\le k_i  &\forall 1 \le i \le t
	\\\cov_i(v) &= \sum_{u \in B(v, r_i)} x_{i, u} &\forall 1 \le i \le t, \forall v \in X
	\\\cov(v) &= \min \LR{\sum_{i = 1}^t \cov_i(v), 1} &\forall v \in X
	\\x_{i, u} &\ge 0 &\forall 1 \le i \le t, \forall u \in X.
\end{align*}

Let $\F$ denote the set of all tuples of subsets $(S_1, \ldots, S_t)$, where $|S_i| \le k_i$ for $1 \le i \le t$. For $v \in X$, and $1 \le i \le t$, we say that $(S_1, \ldots,  S_t) \in \F$ covers $v$ with radius $r_i$, if $d(v, S_i) \le r_i$. Let $\F_i(v) \subseteq \F$ denote the subset of solutions that cover $v$ with radius $r_i$ -- where, the sets $\F_i(v)$ of solutions are assumed to be disjoint by including a solution in $\F_i(v)$ of the smallest index $i$, if it appears in multiple such sets. 

If the instance $\I$ is feasible, then the integer hull of the coverages, $\polycov$ as given below, must be non-empty.
\begin{align*}
	\polycov:
	\\ \sum_{v \in X} \sum_{i \in [t]} \cov_i(v) &\ge m
	\\ \sum_{S \in \F_i(v)} z_S &= \cov_i(v)  &\forall i \in [t], \forall v \in X
	\\ \sum_{S \in \F} z_S &= 1 
	\\ z_S &\ge 0 &\forall S \in \F
\end{align*}

Next, we give a few definitions from \cite{ChakrabartyN21}, generalized to arbitrary $t \ge 2$, for the sake of completeness. These definitions are used in the round-or-cut framework that reduces an instance of Robust $t$-\nukc to Well-Separated Robust $t$-\nukc, as described in Section \ref{sec:t-to-robust}.

\paragraph{$t$-Firefighter Problem.}
The input is a collection of height-$t$ trees, where $L_1$ is the set of roots, and for any $v \in L_i$ with $i \ge 1$, $a_j(v)$ represents the ancestor of $v$ that belongs to $L_j$, where $1 \le j \le i$ ($a_i(v) = v$).
Furthermore, let $w: L_t \to \mathbb{N}$ be a weight function on the leaves. For a root $u \in L_1$, we use $\mathsf{Leaf}(u)$ to denote the set of leaves, i.e., nodes in $L_t$ in the tree rooted at $u$. 

Note that the $\{ \mathsf{Leaf}(u): u \in L_1 \}$ partitions $L_t$. Thus, $((L_1, \ldots, L_t), (a_1, a_2, \ldots, a_t), \mathsf{Leaf}, w)$ completely describes the structure of the tree, where $a_i(v): \bigcup_{i \le j \le t} L_j \to L_i$ is an ancestor function as defined above. Now we define the $t$-FF problem.

\begin{definition}[$t$-FF Problem] \label{def:t-ff}
	Given height-$t$ trees $( \mathcal{T} = (L_1, \ldots, L_t), (a_1, \ldots, a_t), \mathsf{Leaf}, w)$, along with budgets $(k_1, \ldots, k_t)$, we say that $T = (T_1, \ldots, T_t)$, with $T_i \subseteq L_i$ is a feasible solution, if $|T_i| \le k_i$ for $1 \le i \le t$. Let $\mathcal{C}(T) = \{ v \in L_t : a_i(v) \in T_i \text{ for some $1 \le i \le t$} \}$ be the set of leaves covered by the solution. Then, the objective is to find a feasible solution maximizing the weight of the leaves covered. This instance is represented as $\mathcal{I} = (\mathcal{T} = ((L_1, \ldots, L_t), (a_1, \ldots, a_t), \mathsf{Leaf}, w), (k_1, \ldots, k_t))$,
\end{definition}

\begin{definition}[The solution $y$] \label{def:soln-y}
	Given $\cov$, and a collection $\T$ of rooted trees, let $L_1$ denote the set of roots, and let $L_i$, $i > 1$ denote the set of vertices at $j$-th level. Furthermore, for any node $v \in L_i$ with $i > 1$, let $a_j(v)$ denote the ancestor of $v$ that belongs to $L_j$, where $1 \le j < i$. Then, the solution $y$ is defined as follows.
	$$y(v) = \begin{cases}
		\cov_1(v) & \text{ if } v \in L_1
		\\ \min \LR{ \cov_i(v), 1 - \sum_{j < i} \cov_j(a_j(v)) } & \text{ if } v \in L_i, i > 1
	\end{cases}
	$$
\end{definition}

\begin{definition}[The Sparse LP] \label{def:sparse-lp}
	\begin{align*}
		\max \sum_{v \in L_t} w(v) Y(v) &
		\\\sum_{u \in L_1} y_u &\le k_1 - t
		\\\sum_{u \in L_i} y_u &\le k_i  \qquad \forall 2 \le i \le t
		\\Y(v) \coloneqq y_v &+ \sum_{i = 1}^{t-1} y_{a_i(v)} \qquad \forall v \in L_{t}
	\end{align*}
\end{definition}

We now describe two subroutines that are used in the Reduction from Robust $t$-\nukc to Well-Separated Robust $t$-\nukc. We use the same notation and convention as in \cite{ChakrabartyN21}. These two algorithms (Algorithm \ref{alg:hs} and Algorithm \ref{alg:cgk}) are named after Hochbaum, and Shmoys \cite{hochbaumS1985best}; and Chakrabarty, Goyal, and Krishnaswamy \cite{CGK20}, respectively.

\begin{algorithm} 
	\caption{HS$(\text{Metric space} (X, d), r \ge 0, \text{ assignment } \cov: X \to \real^+)$} \label{alg:hs}
	\begin{algorithmic}[1]
		\State $R \gets 0$
		\While{$U \neq \emptyset$}
			\State $u \gets \arg\max_{v \in U} \cov(v) $
			\State $R \gets R \cup \{u\}$
			\State $\Child(u) \gets \{ v \in U : d(u, v) \le r \}$
			\State $U \gets U \setminus \Child(u)$
		\EndWhile
		\State \Return $R, \{ \Child(u) : u \in R \}$.
	\end{algorithmic}
\end{algorithm}

\begin{algorithm} 
	\caption{CGK} \label{alg:cgk}
	\begin{algorithmic}[1]
		\Statex \textbf{Input: } Robust $t$-\nukc instance $\I = ((X, d), (\omega, m), (r_1, \ldots, r_t), (k_1, \ldots, k_t))$, 
		\Statex \qquad\qquad $(\alpha_1, \ldots, \alpha_t)$, where $\alpha_i > 0$ for $1 \le i \le t$, 
		\Statex \qquad\qquad $\cov = (\cov_1, \ldots, \cov_t)$, where each $\cov_i: X \to \real^+$
		\For{$i = t$ \textbf{downto} $1$}
			\State $(L_i, \LR{\Child_i(v) : v \in L_i}) \gets $ HS$((X, d), \alpha_i r_i, \cov'_i \coloneqq \sum_{j = 1}^i \cov_j)$
		\EndFor
		\State \emph{Construct and Return a $t$-FF instance using $\{L_i, \Child_{i} \}_{1 \le i \le t}$ as described below.}
	\end{algorithmic}
\end{algorithm}

We construct the $t$-FF instance based on the sets $L_i$'s constructed, as follows. Consider some $1 \le i \le t-1$, and some $u \in L_i$. Then, for every $v \in \Child_i(u)$, we make $v$ a child of $u$ in a tree $T$. Note that $L_1$ is the set of roots of the trees constructed in this way. Then, we define $\mathsf{Leaf}(u) = \{ v \in L_t : \text{ $v$ is a leaf in the tree rooted at $v$ } \}$, and let $a_i: L_t \to \bigcup_{i \le j \le t} L_j$ be the ancestor function as defined above. Finally, for every $u \in L_t$, let $w(u) = |\mathsf{Child}_t(u)|$. Then, we return the $t$-FF instance $\mathcal{I} = (\mathcal{T} = ((L_1, \ldots, L_t), (a_1, \ldots, a_t), \mathsf{Leaf}, w), (k_1, \ldots, k_t))$.
\end{document}